\documentclass{sig-alternate}
\usepackage{color}
\usepackage{graphicx,caption,subcaption}
\usepackage{amsmath,amssymb,stmaryrd,mathtools}
\usepackage{flushend}
\usepackage{algorithm,algorithmicx,algpseudocode}
\usepackage{hyperref}
\usepackage{acronym}

\graphicspath{{./figures/}}

\newtheorem{remark}{Remark}
\newtheorem{definition}{Definition}
\newtheorem{problem}{Problem}
\newtheorem{theorem}{Theorem}

\newcommand{\ignore}[1]{}

\newcommand{\eg}{e.\,g.,\ }
\newcommand{\ie}{i.\,e.,\ }
\newcommand{\etal}{\emph{et al. }}

\DeclareFontFamily{U}{MnSymbolC}{}
\DeclareSymbolFont{MnSyC}{U}{MnSymbolC}{m}{n}
\DeclareFontShape{U}{MnSymbolC}{m}{n}{
    <-6>  MnSymbolC5
   <6-7>  MnSymbolC6
   <7-8>  MnSymbolC7
   <8-9>  MnSymbolC8
   <9-10> MnSymbolC9
  <10-12> MnSymbolC10
  <12->   MnSymbolC12%
}{}
\DeclareMathSymbol{\powerset}{\mathord}{MnSyC}{180}

\DeclarePairedDelimiter\floor{\lfloor}{\rfloor}

\newcommand{\x}{{\bf x}}
\newcommand{\uu}{{\bf u}}

\newcommand{\G}{{\bf G}}
\newcommand{\F}{{\bf F}}

\newcommand{\U}{{\bf U}}

\newcommand{\uH}{{\bf u}^H}
\newcommand{\yH}{{\bf y}^H}

\newcommand{\n}{\mathcal{N}}

\newcommand{\est}{\lambda_{\alpha_t}}
\newcommand{\pstl}{PrSTL}
\newcommand{\pneg}{\tilde{\neg}}
\newcommand{\ztrue}{p_{\text{t}}^\varphi}
\newcommand{\zfalse}{q_{\text{t}}^\varphi}
\newcommand{\ztruee}{p_{\text{t}}^\psi}
\newcommand{\zfalsee}{q_{\text{t}}^\psi}
\newcommand{\ztruei}{p_{\text{t}}^{\psi_i}}
\newcommand{\zfalsei}{q_{\text{t}}^{\psi_i}}

\begin{document}
%

\title{Safe Control under Uncertainty}

%
%
%
%
%

\numberofauthors{2} 
%
\author{
%
\alignauthor Dorsa Sadigh\\
       \affaddr{UC Berkeley}\\
       \affaddr{Berkeley, CA, USA}\\
       \email{dsadigh@berkeley.edu}
\alignauthor Ashish Kapoor\\
       \affaddr{Microsoft Research}\\
       \affaddr{Redmond, WA, USA}\\
       \email{akapoor@microsoft.com}
}

\maketitle

\begin{abstract}
Controller synthesis for hybrid systems that satisfy temporal specifications expressing various system properties is a challenging problem that has drawn the attention of many researchers. 
However, making the assumption that such temporal properties are deterministic is far from the reality. 
For example, many of the properties the controller has to satisfy are learned through machine learning techniques based on sensor input data. In this paper, we propose a new logic, Probabilistic Signal Temporal Logic (\pstl), as an expressive language to define the stochastic properties, and enforce probabilistic guarantees on them. 
We further show how to synthesize safe controllers using this logic for cyber-physical systems under the assumption that the stochastic properties are based on a set of Gaussian random variables.
One of the key distinguishing features of \pstl~is that the encoded logic is adaptive and changes as the system encounters additional data and updates its beliefs about the latent random variables that define the safety properties. 
We demonstrate our approach by synthesizing safe controllers under the~\pstl~specifications for multiple case studies including control of quadrotors and autonomous vehicles in dynamic environments.
\end{abstract}





\section{Introduction}
\label{sec:intro}
Synthesizing safe controllers for cyber-physical systems (CPS) is a challenging problem, due to various factors that include uncertainty arising from the environment. For example, any safe control strategy for quadcopters need to incorporate predictive information about wind gusts and any associated uncertainty in such predictions. Similarly, in the case of autonomous driving, the controller needs a probabilistic predictive model about the other vehicles on the road in order to avoid collisions. Without a model of uncertainty that would characterize all possible outcomes, there is no guarantee that the synthesized control will be safe.
 
The field of Machine Learning has a rich set of tools that can characterize uncertainties. Specifically, Bayesian graphical models~\cite{jordan1998learning} have been very popular in modeling uncertainties arising in scenarios common to CPS. For example, one of the common strategies in CPS is to build classifiers or predictors based on acquired sensor data. It is appealing to consider such predictive systems in synthesizing safe controllers for dynamical systems. However, it is well known that it is almost impossible to guarantee a prediction system that works perfectly all the time. Consequently, we need to devise control synthesis methodologies that are aware of such limitations imposed by the Machine Learning systems. Specifically, we need to build a framework that is capable of synthesis of safe controllers by being aware of when the prediction system would work and when it would fail.
 
In this paper, we propose a methodology for safe controller synthesis using the novel \emph{Probabilistic Signal Temporal Logic} (PrSTL) that allows us to embed predictive models and their associated uncertainties. The key ingredient of the framework is a logic specification that allows embedding of uncertainties via {\em probabilistic predicates} that take random variables as parameters. These random variables allow incorporation of Bayesian graphical models in these predicates, thereby resulting in a powerful logic specification that can reason about safety under uncertainty. One of the main advantages of using Bayesian graphical models (or Bayesian methods in general) is the fact that the predictions provided are full distributions associated with the quantity of interest as opposed to a point estimate. For example, a classical Machine Learning method might just provide a value for wind speed, however under the Bayesian paradigm we would be recovering an entire probability distribution over all possible winds. Finally, another distinguishing aspect of our framework is that these probabilistic predicates are adaptive: as the system sees more and more data, the inferred distribution over the latent variables of interest can change leading to change in the predicates themselves.
 
Previous efforts for synthesizing safe controllers either operate under deterministic environments or model uncertainty only as part of the dynamics of the system. For example, Signal Temporal Logic (STL)~\cite{maler2004monitoring} provides a framework for expressing real-valued dense-time temporal properties for safety, but assumes that the signal provided from the trajectory of the system is deterministically defined by the system dynamics. Similarly, other approaches that model uncertainty as a variable added to the dynamics~\cite{sadigh2014learning,fu2015computational,svorevnova2015temporal,fu2014integrating,fu2014probably} lack clear connections to various sources of uncertainty present in the environment. Specifically, with prior approaches there is no clear understanding of how uncertainty arising due to sensing and classification could be incorporated while reasoning about safe control trajectories.
 
In this paper we aim to alleviate these issues by defining a probabilistic logical specification framework that has the capacity to reason about safe control strategies by embedding various predictions and their associated uncertainty. Specifically, our contributions in this paper are:
\begin{itemize}
\item Formally define PrSTL, a logic for expressing probabilistic properties that can embed Bayesian graphical models.
\item Formalize a receding horizon control problem to satisfy PrSTL specifications.
\item Provide a novel solution for the controller synthesis problem using Mixed Integer Semi-Definite Programs.
\item Provide a toolbox implementing our algorithms and showcasing experiments in autonomous driving and control of quadrotors.
\end{itemize}
The rest of this paper is organized as follows: In Section~\ref{sec:related} we go over some of the related work in the area of stochastic control, and controller synthesis under safety. In Section~\ref{sec:prelim}, we discuss the preliminaries, and in Section~\ref{sec:probdef}, we define the problem statement along with the formal definition of PrSTL. Section~\ref{sec:examples} illustrates our experimental results, and we conclude in Section~\ref{sec:dis}.

\section{Related Work}
\label{sec:related}
Over the years researchers have proposed different approaches for safe control of cyber-physical systems.
For instance, designing controllers under reachability analysis is a well-studied method that allows specifying safety and reachability properties~\cite{mitchell2000level,mitchell2005time}.
More recently, safe learning approaches construct controllers that keep the system in the safe region, while the optimal strategy is learned online~\cite{gillula2012guaranteed,aswani2013provably,akametalu2014reachability}. However, finding the reachable set is computationally expensive, which makes this approach impractical for most interesting cyber-physical systems.
Controller synthesis under temporal specifications such as Linear Temporal Logic (LTL) allows expressing more interesting properties of the system and environment,~\eg safety, liveness, response, stability, etc., and has shown promising results in robotics applications~\cite{piterman2006synthesis,kress2009temporal,livingston2012backtracking,wongpiromsarn2010receding,gol2015temporal}.
However, synthesis for LTL requires time and space discretization, which again suffers from the curse of dimensionality. Although, this approach is effective at high level planning, it is unsuitable for synthesizing control inputs at the level of dynamical systems.
More recently, Raman~\etal have studied synthesis for Signal Temporal Logic (STL), which allows real-valued, dense-time properties in a receding horizon setting~\cite{raman2014model,Raman15}. 
Although, this approach requires solving mixed integer linear programs, it has shown promising results in practice.
One downside of specifying properties in STL or LTL is that the properties of the system and environment have to be expressed deterministically.
Knowledge of the exact parameters and bounds of the specification is an unrealistic assumption for most CPS applications, where the system interacts with uncertain environments, and has partial knowledge of the world based on its sensors and classifiers.

The problem of controller synthesis under uncertainty is also a well-studied topic.
One of the most effective approaches in robust control under uncertainty is modeling the environment uncertainty as part of the dynamics of the system, and finding the optimal strategy for the worst case disturbance~\cite{kothare1996robust,wang1992robust,zhou1998essentials}.
However, considering worst case environment is inapplicable and too conservative.
More recently, researchers have proposed modeling the environment in a chance constrained framework, and there are some promising results in the area of urban autonomous driving~\cite{lenz2015stochastic,vitus2012stochastic,vitus2013probabilistic,blackmore2011chance,carvalhoa2014stochastic}.
In most previous work the uncertainty from the system or environment is modeled as part of the dynamics, and there is not an intuitive connection between the properties, and the sensing and classification capabilities of the system.
In addition, there has been efforts in verification and synthesis of controllers for temporal properties given probabilistic transition systems~\cite{sadigh2014learning,fu2015computational,svorevnova2015temporal,fu2014integrating,fu2014probably,puggelli2013polynomial}.
To best of our knowledge, none of the previous studies consider scenarios, where the uncertainty and confidence in properties is originated from classifiers rather than the dynamics of the system, and is expressed as part of the specification.
In this work, we propose a more natural framework for expressing temporal and Boolean properties over different sources of uncertainty, and their interconnect, which allows synthesizing safe controllers while considering such probabilistic temporal specifications.

\section{Preliminaries}
\label{sec:prelim}
\subsection{Hybrid Dynamical System}
We consider a continuous time hybrid dynamical system:
\begin{equation}
\begin{aligned}
	&\dot{x_t} = f(x_t,u_t)\\
	&y_t = g(x_t,u_t).
\end{aligned}	
\end{equation}
Here, $x_t \in \mathcal{X} \subseteq (\mathbb{R}^{n_c} \times \{0,1\}^{n_d})$ is a signal representing the continuous and discrete mode of the system at time $t$, $u_t \in U\subseteq (\mathbb{R}^{m_c} \times \{0,1\}^{m_d})$ is the control input and $y_t \in Y\subseteq (\mathbb{R}^{p_c} \times \{0,1\}^{p_d})$ is the output of the system at time $t$.
This continuous system can be discretized using time intervals $dt>0$, and every discrete time step is $k = \floor{t/ dt}$. The discrete time hybrid dynamical system is formalized as:
\begin{equation}
\label{eq:dynamics}
\begin{aligned}
	&x_{k+1} = f_d(x_k,u_k)\\
	&y_k = g_d(x_k,u_k).
\end{aligned}	
\end{equation}
We let $x_0 \in \mathcal{X}$ denote the initial state of the dynamical system. We express an infinite \emph{run} of the system as: $\xi = (x_0,u_0),(x_1,u_1),\dotsc$. Given the initial state $x_0$, and a finite length input sequence: ${\bf u}^H = u_0, u_1, \dotsc, u_{H-1}$, the finite horizon \emph{run} or \emph{trajectory} of the system following the dynamics in equation~\eqref{eq:dynamics} is:
\begin{equation}
\xi^H(x_0,{\bf u}^H) = (x_0, u_0),(x_1,u_1), \dotsc, (x_H, u_H).
\end{equation}
Furthermore, we let $\xi(t) = (x_t,u_t)$ be a \emph{signal} consisting of the state and input of the system at time $t$; $\xi_x(t) = x_t$ is the state, and $\xi_u(t) = u_t$ is the input at time $t$.

The output of the system is also computed to be $\yH = y_0,y_1, \dotsc, y_{H-1}$.
A cost function is defined for the finite horizon trajectory, denoted by $J(\xi^H)$, and maps $\xi^H \in \Xi$, the set of all trajectories to positive real valued costs in $\mathbb{R}^+$.
\subsection{Controller Synthesis for Signal Temporal Logic}
\label{sec:stl}
\emph{Signal Temporal Logic (STL)} is an expressive framework that allows reasoning about real-valued dense-time functions, and has been largely used for defining robustness measures and monitoring properties of real-time signals of hybrid systems~\cite{maler2004monitoring,donze2010robust,donze2013efficient}. More recently there has been interest in synthesizing controllers that satisfy STL properties~\cite{raman2014model,Raman15}.

Formally, $(\xi,t) \models \varphi$ denotes that a signal $\xi$ satisfies the STL formula $\varphi$ at time $t$.
An atomic predicate of an STL formula is represented by inequalities of the form $\mu(\xi(t)) > 0$, where $\mu$ is a function of the signal $\xi$ at time $t$. The truth value of the predicate $\mu$ is equivalent to $\mu(\xi(t)) >0$.
Any STL formula consists of Boolean and temporal operations on these predicates and the syntax of STL formulae $\varphi$ is defined recursively as follows:
\begin{equation}
\varphi ::= \mu \:|\: \neg \mu \:|\: \varphi \wedge \psi \:|\: \G_{[a,b]} \psi \:|\: \varphi \: \U_{[a,b]} \psi,
\end{equation}
where $\psi$ and $\varphi$ are STL formulae, $\G$ denotes the \emph{globally} operator and $\U$ is the \emph{until} operator.
For instance, $\xi \models \G_{[a,b]} \psi$ specifies that $\psi$ must hold at all times in the given interval, $t\in [a,b]$ of signal $\xi$. We can also define $\F$ the \emph{eventually} operator, and $\F_{[a,b]} \psi =\neg \G_{[a,b]} \neg \psi$.
Satisfaction of an STL formula $\varphi$ for a signal $\xi$ at time $t$ is formally defined as follows:
\small
\begin{equation}
\label{eq:STL}
\centering
\begin{array}{lll}
(\xi,t) \models \mu &\Leftrightarrow & \mu(\xi(t)) > 0\\
(\xi,t) \models \neg \mu &\Leftrightarrow & \neg((\xi,t)\models \mu)\\
(\xi,t) \models \varphi \land \psi &\Leftrightarrow & (\xi,t) \models \varphi \land (\xi,t) \models \psi\\
(\xi,t) \models \varphi \lor \psi &\Leftrightarrow & (\xi,t) \models \varphi \lor (\xi,t) \models \psi\\
(\xi,t) \models \G_{[a,b]} \varphi &\Leftrightarrow & \forall t'\in [t+a, t+b], (\xi,t') \models \varphi\\
(\xi,t) \models \F_{[a,b]} \varphi & \Leftrightarrow & \exists t'\in [t+a, t+b], (\xi,t') \models \varphi\\
(\xi,t) \models \varphi~\U_{[a,b]}~\psi &\Leftrightarrow & \exists t' \in [t+a,t+b] \mbox{ s.t. } (\xi,t') \models \psi \\
&&\land \forall t'' \in [t,t'], (\xi,t'') \models \varphi.
\end{array}
\end{equation}
\normalsize
An STL formula $\varphi$ is \emph{bounded-time} if it contains no unbounded operators.
The \emph{bound} of a formula is defined as the maximum over the sum of all nested upperbounds on the STL formulae.

Synthesizing controllers that satisfy STL properties is a non-trivial task. Most promising approaches are based on Receding Horizon Control or Model Predictive Control (MPC) ~\cite{morari1993model} that aim to iteratively optimize a cost function $J(\xi^H)$ of interest. Specifically, starting with an initial state $x_0$, the MPC scheme aims to determine the optimal control strategy $\uH$ given the dynamics model of the system as in equation~\eqref{eq:dynamics}, while satisfying the STL formula $\varphi$. The constraints represented using STL allow expression of temporal specifications on the runs of the system and environment and limit the allowed behavior of the closed loop system ~\cite{Raman15,raman2014model}.

Prior work shows that MPC optimization with STL constraints $\varphi$  can be posed as a \emph{Mixed Integer Linear Program} (MILP)~\cite{raman2014model,griva2009linear}. It is well-known that the global optimality of this approach is not guaranteed; nonetheless, MPC is fairly used in practice, and has shown to perform well.


\subsection{Bayesian Methods to Model Uncertainty}
\label{sec:Bayes}
Probability theory provides a natural way to represent uncertainty in the environment and recent advances in Machine Learning and Perception have heavily relied on Bayesian methods to infer distributions over latent phenomenon of interest~\cite{gelman2014bayesian,jordan1998learning}. The two key ingredients include (a) Bayesian networks (equivalently graphical models) that allow expression of complex interactions between sets of latent variables and (b) the Bayesian inference procedure that numerically computes probability distributions over the variables of interest. One of the key distinguishing aspects of the Bayesian methodology is that, unlike other optimization based machine learning methods, the entire distributions over the variables of interest are available. Such distributions completely characterize the uncertainty present in the system and are crucial for our goal of synthesizing safe controllers.
 
While a thorough discussion of Bayesian networks and associated methods to model uncertainty is beyond the scope of this paper, we highlight these methods on the task of inferring classifiers from observed training data. Formally, given a set of training data points ${\bf X}_L = \{{\bf x}_1, \dotsc , {\bf x}_n\}$, with observations ${\bf t}_L = \{t_1, \dotsc , t_n\}$, where $t_i \in \{+1, -1\}$,
we are interested in finding a hyperplane ${\bf w}$ that separates the points belonging to the two classes according to $\mbox{sgn}({\bf w}^T {\bf x})$. Under the Bayesian paradigm, we are interested in the distribution:
\begin{equation}
\label{eq:bayes}
\begin{aligned}
p({\bf w} | {\bf X}_L, {\bf t}_L) &= p({\bf w})\cdot p({\bf t}_L | {\bf X}_L, {\bf w})\\
     &= p({\bf w}) \prod_{i}{p(t_i | {\bf w}, {\bf x}_i)}\\
     &= p({\bf w}) \prod_{i}{\mathbb{I}\:[\mbox{sgn}({\bf w}^T{\bf x}_i) = t_i]}.
\end{aligned}
\end{equation}
The first line in the above equation stems from the Bayes rule, and the second line simply exploits the fact that given the classifier ${\bf w}$ the labels for each of the points in the data set are independent. The expression $\mathbb{I}[\cdot]$ in the third line is an indicator function which evaluates to $1$ when the condition inside the brackets holds. Thus, equation~\eqref{eq:bayes} starts from a prior $p(\bf w)$ over the classifiers and eventually by incorporating the training data points, infers a posterior distribution over the set of all the classifiers that respect the observed labels and the points. While the above equation expresses the statistical dependencies among the various variables (\ie the model), there are various Bayesian inference techniques~\cite{minka2001family,beal2003variational,andrieu2003introduction} that would allow numerical computation of the posterior distribution $p({\bf w} | {\bf X}_L, {\bf t}_L)$ of interest. In the above case of Bayesian classifier, the popular method of choice is to use Expectation Propagation~\cite{minka2001family} to infer $p({\bf w} | {\bf X}_L, {\bf t}_L)$ as a Gaussian distribution $N({\bf w}; \bar{\bf w}, \Sigma)$. Linear application of this classifier to a data point as ${\bf w}^T{\bf x}$ results in a Gaussian distribution of the prediction with the mean ${\bf w}^T{\bf x}$ and the variance ${\bf x}^T\Sigma{\bf x}$. Similarly, for the case of Bayesian linear regression the same procedure can be followed, albeit with continuous target variables $t \in \mathbb{R} $.
 
Note that these Bayesian linear classifiers and regressors are a fairly rich class of models and have similar or better representation capabilities as kernel machines~\cite{williams2006gaussian}. In this work, we specifically aim to incorporate such rich family of classification models in safe controller synthesis.

\section{Problem Statement}
\label{sec:probdef}
We propose \emph{Probabilistic Signal Temporal Logic} (\pstl) that allows us to express uncertainty over the latent variables via probabilistic specifications. The key idea in our work is to first incorporate random variables in predicates, and then express
temporal and Boolean operations on such predicates. The proposed logic provides an expressive framework for defining safety conditions
under a wide variety of uncertainties, including the uncertainty that arises due to application of Machine Learning classifiers.

The core ingredient in this work is the insight that when the uncertainty over the random variable is reasoned out in a Bayesian framework, we can use the inferred probability distributions to efficiently derive constraints from the \pstl~specifications.
We provide a novel solution for synthesizing controllers for dynamical systems given different \pstl~properties. An interesting aspect of this framework is that the \pstl~formulae can evolve at every step. For example, a classifier associated with the dynamical system can continue to learn with time, thereby changing the inferred probability distributions on the latent random variables.

\subsection{Probabilistic Signal Temporal Logic}
\label{sec:pstl}
\pstl~supports probabilistic temporal properties on real-valued, dense-time signals. Specifically, $(\xi,t) \models \varphi$ denotes that the signal $\xi$ satisfies the \pstl~formula $\varphi$ at time $t$. We introduce the notion of a probabilistic atomic predicate $\est(\xi(t))$ of a \pstl~formula that is parameterized with a time-varying random variable $\alpha_t$ drawn from a distribution $p(\alpha_t)$ at every time step:
\begin{equation}
\label{eq:pred}
(\xi,t) \models \est^{\epsilon_t} \: \Longleftrightarrow \: P(\est(\xi(t)) < 0) > 1-\epsilon_t.
\end{equation}
Here $P(\cdot)$ represents the probability of the event and $1-\epsilon_t$ defines the \emph{tolerance} level in satisfaction of the probabilistic properties. The parameter $\epsilon_t \in[0, 1]$ is a small time-varying positive number and represents the threshold on satisfaction probability of $\est (\xi(t)) < 0$. A signal $\xi(t)$ satisfies the PrSTL predicate $\est$ with confidence $1-\epsilon_t$ if and only if:
\begin{equation}
\label{eq:marginal}
\int_{\alpha_t}{\mathbb{I}[\est(\xi(t))<0]\: p(\alpha_t)\: d\alpha_t} > 1-\epsilon_t.
\end{equation}
Here $\mathbb{I}[\cdot]$ is an indicator function, and the equation marginalizes out the random variable $\alpha_t$ with the probability density $p(\alpha_t)$. The truth value of the \pstl~predicate $\lambda_{\alpha_t}^{\epsilon_t}$ thus is equivalent to satisfaction of the probabilistic constraint in equation~\eqref{eq:pred}. We would like to point out that computing such integrals for general distributions is computationally difficult; however, there are many parameterized distributions (e.g., Gaussian and other members of the exponential family) for which there exists either a closed form solution or efficient numerical procedures.

Note that this probabilistic atomic predicate $\est(\xi(t))$ is a stochastic function of the signal $\xi$ at time $t$ and expresses a model of the uncertainty in environment based on the observed signals. As the system evolves and observes more data about the environment, the distribution over the random variable $\alpha_t$ changes over time, thereby leading to an adaptive \pstl~predicate.
The \pstl~formula consists of Boolean and temporal operations over their predicates. We formulate \pstl~in negation normal form, and recursively define the syntax of the logic as:
\begin{equation}
\label{eq:syntax}
\varphi ::=\lambda_{\alpha_t}^{\epsilon_t} \:|\: \pneg \lambda_{\alpha_t}^{\epsilon_t}  \:|\: \varphi \wedge \psi \:|\: \G_{[a,b]} \psi \:|\: \varphi \: \U_{[a,b]} \psi.
\end{equation}

Here, $\varphi$ is a \pstl~formula, which is built upon predicates $\lambda_{\alpha_t}^{\epsilon_t}$ defined in equation~\eqref{eq:pred}, \emph{propositional formulae} $\varphi$ composed of the predicates and Boolean operators such as $\wedge$ (and), $\pneg$ (negation), and \emph{temporal operators} on $\varphi$ such as $\G$ (globally), $\F$ (eventually) and $\U$ (until).
Note, that in these operations the \pstl~predicates can have different probabilistic parameters,~\ie $\alpha_t$ and $\epsilon_t$.
In addition, satisfaction of the \pstl~formulae for each of the Boolean and temporal operations based on the predicates is defined as:
\small
\begin{equation}
\label{eq:PSTL}
\begin{array}{lll}
(\xi,t) \models \lambda_{\alpha_t}^{\epsilon_t} &\!\Leftrightarrow & P(\lambda_{\alpha_t}(\xi(t)) < 0) > 1- \epsilon_t\\
(\xi,t) \models \pneg \lambda_{\alpha_t}^{\epsilon_t} &\!\Leftrightarrow & (\xi,t) \models -\lambda^{\epsilon_t}_{\alpha_t}\\
(\xi,t) \models \varphi \land \psi &\Leftrightarrow & (\xi,t) \models \varphi \land (\xi,t) \models \psi\\
(\xi,t) \models \varphi \lor \psi &\Leftrightarrow & (\xi,t) \models \varphi \lor (\xi,t) \models \psi\\
(\xi,t) \models \G_{[a,b]} \varphi &\Leftrightarrow & \forall t'\in [t+a, t+b], (\xi,t') \models \varphi\\
(\xi,t) \models \F_{[a,b]} \varphi & \Leftrightarrow & \exists t'\in [t+a, t+b], (\xi,t') \models \varphi\\
(\xi,t) \models \varphi~\U_{[a,b]}~\psi &\Leftrightarrow & \exists t' \in [t+a,t+b] \mbox{ s.t. } (\xi,t') \models \psi \\
&&\land \forall t'' \in [t,t'], (\xi,t'') \models \varphi.
\end{array}
\end{equation}
\normalsize
\begin{remark}
Note that $\pneg$ (negation) defined above, does not follow the traditional logical complement properties,~\ie  a formula and its negation can both be satisfied or violated by our definition of negation.
Satisfaction of a complement of a \pstl~formula is equivalent to negating the formula's function $-\lambda_{\alpha_t}^{\epsilon_t}$.
\end{remark}

\begin{remark}
The \pstl~framework reduces to STL, when the distribution $p(\alpha_t)$ is a Dirac distribution. A Dirac or a point distribution over $\alpha_t$ enforces $\lambda_{\alpha_t}(\xi(t))<0$ to be deterministic and equivalent to an STL predicate $\mu$ defined in Section~\ref{sec:stl}.
\end{remark}
\subsection{Controller Synthesis for Probabilistic Signal Temporal Logic}
We now formally define the controller synthesis problem in the MPC framework with \pstl~specifications.
\begin{problem}
\label{prob:1}	
Given a hybrid dynamical system as in equation~\eqref{eq:dynamics}, an initial state $x_0$, a \pstl~formula $\varphi$,
an MPC horizon $H$, and a cost function $J(\xi^H)$ defined for a finite horizon trajectory $\xi^H$ find:
\begin{equation}
\label{eq:prob1}
\begin{aligned}
	 &\underset{\uH}{\text{argmin}}
	& \quad & J(\xi^H(x_0, \uH))\\
	 &\text{subject to}
	& \quad & \xi^H(x_0,\uH) \models \varphi.
\end{aligned}	
\end{equation}
\end{problem}
Problem~\eqref{prob:1} formulates a framework for finding a control strategy $\uH$ that optimizes a given cost function, and satisfies a \pstl~formula. Finding the best strategy for this optimization given only deterministic \pstl~formulae, where $\alpha_t$ is drawn from a Dirac distribution is the same as solving a set of mixed integer linear constraints. In this section, we show how the optimization can be solved for the general case of \pstl~by translating the formula to a set of mixed integer constraints. Specifically, we provide full solution for the Gaussian distributions in Problem~\ref{prob:1}, where the optimization reduces to mixed integer semi-definite programs.
\subsubsection{Mixed Integer Constraints}
\label{sec:integer}
We first discuss how every \pstl~formula generates a set of integer constraints. Given a \pstl~formula, we introduce two integer variables for every time step $t$, $\ztrue$ and $\zfalse \in \{ 0,1 \}$, which correspond to the truth value of the \pstl~formula and its negation respectively. These variables enforce satisfaction of the \pstl~formula $\varphi$ as follows:
\begin{equation}
\begin{aligned}
&\ztrue = 1 &\implies & (\xi,t) \models \varphi \\
&\zfalse = 1 &\implies & (\xi,t) \models \pneg \varphi
\end{aligned}
\end{equation}

The formula $\varphi$ holds true if $\ztrue = 1$, and its negation $\pneg \varphi$ (defined in Section~\ref{sec:pstl}) holds true if $\zfalse = 1$.
Due to our definition of negation for probabilistic formulae, there exist signals for which $\ztrue$, and $\zfalse$ can both be set to 1, where both $\varphi$, and $\pneg \varphi$ are satisfied by the signal. This explains the construction of two integer variables for every formula. Using both integer variables, we define the constraints required for logical and temporal operations of~\pstl~on $\ztrue$ and $\zfalse$ for all times. These integer variables enforce the truth value of the formula $\varphi$, and we refer to them as \emph{truth value enforcers}:
\begin{itemize}
\item {\bf Negation} $ (\varphi = \pneg \psi): \quad \ztrue \leq \zfalsee$ and  $\zfalse \leq \ztruee$
\item {\bf Conjunction} $ (\varphi = \wedge_{i=1}^N \psi_i): \quad \ztrue \leq \ztruei$ and \\ $\zfalse \leq \sum_{i=1}^N \zfalsei$
\item {\bf Disjunction} $ (\varphi = \vee_{i=1}^N \psi_i):$ $\varphi = \pneg \wedge_{i=1}^N \pneg \psi_i$
\item {\bf Globally} $(\varphi = \G_{[a,b]}\psi):$
	\begin{eqnarray*}
		& \ztrue \leq p_{t'}^\psi & \forall t' \in [t+a, \:\min(t+b,H\!-\!1)], \\
		& \zfalse \leq \sum_{t'=t+a}^{t+b}q_{t'}^\psi & \text{(Only for $t<H-b$).}
	\end{eqnarray*}
	\label{eq:global}
\item {\bf Eventually} $(\varphi = \F_{[a,b]}\psi):$ $\varphi = \pneg\: \G_{[a,b]} \pneg\psi$.
\item {\bf Unbounded Until} $(\varphi =\psi_1 \:\tilde{\U}_{[0,\infty)}\psi_2):$\\
$\bigvee_{t=0}^{H-1} \big ((\G_{[0,t]} \psi_1) \wedge (\G_{[t,t]} \psi_2) \big )\vee \G_{[0,H-1]} \psi_1 $
\item {\bf Bounded Until} $(\varphi = \psi_1\: \U_{[a,b]}\psi_2):$ $\varphi = \G_{[0,a]}\psi_1 \wedge \F_{[a,b]}\psi_2 \wedge \G_{[a,a]}(\psi_1 \tilde{\U}_{[0,\infty)} \psi_2)$
\end{itemize}

Here, we have shown how $\ztrue$ and $\zfalse$ are defined for every logical property such as \emph{negation}, \emph{conjunction}, and \emph{disjunction}, and every temporal property such as \emph{globally}, \emph{eventually}, and \emph{until}. We use $\tilde{\U}$ to refer to unbounded until, and $\U$ for bounded until.

Note that while synthesizing controllers for \pstl~formulae in an MPC scheme, we sometimes are required to evaluate satisfaction of the formula outside of the horizon range $H$. For instance, a property $\G_{[a,b]}\varphi$ might need to be evaluated beyond $H$ for some $t'\in [t+a,t+b]$. In such cases, our proposal is to act optimistically, which means that we assume the formula holds true for the time steps outside of the horizon of \emph{globally} operator, and similarly assume the formula does not hold true for the negation of the \emph{globally} operator. This optimism is evident in formulating the truth value enforcers of the \emph{globally} operator above, and based on that, it is specified for other temporal properties.

Based on the recursive definition of \pstl, and the above encoding, the truth value enforcers of every \pstl~formula is defined using a set of integer inequalities involving a composition of the truth value enforcers of the inner predicates.
\subsubsection{Satisfaction of \pstl~predicates}
\label{sec:pred}
We have defined the \pstl~predicate $\lambda_{\alpha_t}^{\epsilon_t}$ for a general function, $\lambda_{\alpha_t}(\xi(t))$ of the signal $\xi$ at time $t$. In general, the function allows a random variable $\alpha_t \sim p(\alpha_t)$ drawn from any distribution at every time step. The general problem of controller synthesis that would satisfy the PrSTL predicates is computationally difficult due to the fact that evaluation of the predicates boils down to computing an integration depicted in equation \eqref{eq:marginal}. Consequently, in order to solve the control problem in equation~\eqref{eq:prob1} we need to enforce a structure on the predicates of $\varphi$. In this section, we explore the linear-Gaussian structure of the predicates that appear in many of the real-world cases and show how they translate into Mixed Integer SDPs.

Formally, if $\varphi = \lambda_{\alpha_t}^{\epsilon_t}$ is only a single predicate, the optimization in equation~\eqref{eq:prob1} will reduce to:
\begin{equation}
\label{eq:prob3}
\begin{aligned}
	 &\underset{\uH}{\text{argmin}}
	&\quad&  J(\xi^H(x_0, \uH))&\\
	 &\text{subject to}
	&\quad&  (\xi,t) \models \lambda_{\alpha_t}^{\epsilon_t} & \forall t \in \{0,\dotsc, H\!-\!1\}.
\end{aligned}	
\end{equation}
This optimization translates to a chance constrained problem ~\cite{ben2009robust,boyd2004convex,vitus2012stochastic,vitus2013probabilistic,lenz2015stochastic,blackmore2011chance} at every time step of the horizon, based on the definition of \pstl~predicates in equation~\eqref{eq:pred}:
\begin{equation}
\label{eq:prob4}
\begin{aligned}
	 &\underset{\uH}{\text{argmin}}
	&\quad &J(\xi^H(x_0, \uH))& \\
	 &\text{subject to}
	&\quad& P(\est(\xi(t)) < 0) > 1-\epsilon_t &\\
	&&& \forall t \in \{0,\dotsc, H\!-\!1\}.&
\end{aligned}	
\end{equation}
One of the big challenges with such chance constrained optimization is there are no guarantees that the above optimization in equation~\eqref{eq:prob4} is convex. The convexity of the problem depends on the structure of the function $\lambda_{\alpha_t}$, and the distribution $p(\alpha_t)$.

It turns out that the problem takes a particularly simple convex form when the function $\lambda_{\alpha_t}$ takes a linear-Gaussian form,~\ie the random variable $\alpha_t$ comes from a Gaussian distribution and the function itself is linear in ${\alpha_t}$:
\begin{equation}
\lambda_{\alpha_t}(\xi(t)) = {\alpha_t}^\top \xi_x(t) = {\alpha_t}^\top x_t, \quad \alpha_t \sim \n(\mu_t,\Sigma_t).
\end{equation}

It is easy to show that for this structure of $\lambda_{\alpha_t}$, where $\lambda_{\alpha_t}$ is a weighted sum of the states with Gaussian weights $\alpha_t$, the chance constrained optimization in equation~\eqref{eq:prob4} is convex~\cite{van1963minimum,kataoka1963stochastic}. Specifically, the optimization problem can be transformed to a second-order cone program (SOCP). To see this, we consider normally distributed random variable $\nu \sim \n(0,1)$, its cumulative distribution function (CDF) $\Phi$:
\begin{equation}
\Phi (z) =  \int_{-\infty}^z \frac{1}{\sqrt{2\pi}}e^{\frac{-t^2}{2}}dt.
\end{equation}
Then, the chance constrained optimization reduces to SOCP via the following derivation:
\begin{equation}
\begin{aligned}
	&P(\lambda_{\alpha_t}(\xi(t)) < 0) > 1-{\epsilon_t}\\
	&\iff P(\alpha_t^\top x_t < 0) > 1-\epsilon_t \\
	&\iff P(\nu < \frac{-\mu_t^\top x_t}{x_t^\top \Sigma_t x_t}) > 1-\epsilon_t \\
	&\iff \int_{-\infty}^{\frac{-\mu_t^\top x_t}{x_t^\top \Sigma_t x_t}} \frac{1}{\sqrt{2\pi}} e^{\frac{-t^2}{2}} dt > 1-\epsilon_t \\
	&\iff \Phi(\frac{\mu_t^\top x_t}{x_t^\top \Sigma_t x_t}) < \epsilon_t \\
	&\iff \mu_t^\top x_t - \Phi^{-1}(\epsilon_t) ||\Sigma_t^{1/2}x_t||_2  < 0
\end{aligned}	
\end{equation}
 
In this formulation, $\mu_t^\top x_t$ is the linear term, where $\mu_t$ is the mean of the random variable $\alpha_t$ at every time step, and $||\Sigma_t^{1/2}x_t||_2$ is the $l_2$-norm representing a quadratic term, where $\Sigma_t$ is the variance of $\alpha_t$. This quadratic term is scaled by $\Phi^{-1}(\epsilon_t)$, the inverse of the Normal CDF function, which is negative for small values of $\epsilon_t \leq 0.5$. Thus, every chance constraint can be reformulated as a SOCP, and as a result with a convex cost function $J(\xi^H)$, we can efficiently solve the following convex optimization for every predicate of \pstl:
\begin{equation}
\begin{aligned}
	 &\underset{\uH}{\text{minimize}}
	& \quad &J(\xi^H(x_0, \uH))&\\
	 &\text{subject to}
	&\quad  &\mu_t^\top x_t - \Phi^{-1}(\epsilon_t) ||\Sigma_t^{1/2} x_t||_2< 0 &\\
	& & &\forall t\in\{0,\dotsc, H-1\}. &
	\end{aligned}		
\end{equation}

Assuming the a linear-Gaussian form of the function, we generate the SOCP above and easily translate it to a semi-definite program (SDP) by introducing auxiliary variables~\cite{boyd2004convex}. We can use this semi-definite program that solves the problem in equation~\eqref{eq:prob3} with a single constraint $\varphi = \lambda_{\alpha_t}^{\epsilon_t}$ as a building block, and use it multiple times to handle complex \pstl~formulae. Specifically, any \pstl~formula can be decomposed to its predicates by recursively introducing integer variables that correspond to the truth value enforcers of the formula at every step as discussed in Section~\ref{sec:integer}.

We would like to point out that assuming linear-Gaussian form of the function $\lambda_{\alpha_t}$ is not too restrictive. The linear-Gaussian form subsumes the case of Bayesian linear classifiers, and consequently the framework can be applied to a wide variety of scenarios where a classification or regression function needs to estimate quantities of interest that are critical for safety. Furthermore, the framework is applicable to all random variables whose distributions exhibit unimodal behavior and aligned with the large law of numbers. Finally, for the cases of non-Gaussian random variables, there are many approximate inference procedures that can approximate the distributions as Gaussian distributions effectively.

\subsubsection{Convex Subset of \pstl}
As discussed in the previous section~\ref{sec:pred}, at the predicate level of $\varphi$, we create a chance constrained problem for predicates $\lambda_{\alpha_t}^{\epsilon_t}$. These predicates of the \pstl~formulae can be reformulated as a semi-definite program, where the predicates are over intersections of cone of positive definite matrices with affine spaces.
Semi-definite programs are special cases of convex optimization; consequently, solving Problem~\ref{prob:1}, only for \pstl~predicates is a convex optimization problem. Note that in Section~\ref{sec:integer} we introduced integer variables for temporal and Boolean operators of the \pstl~formula. Construction of such integer variables increases the complexity of Problem~\ref{prob:1}, and results in a mixed integer semi-definite program (MISDP). However, we are not always required to create integer variables for all temporal and Boolean operators. Therefore, we define \emph{Convex \pstl} as a subset of \pstl~formulae that can be solved without constructing integer variables.

\begin{definition}
Convex \pstl~is a subset of PrSTL such that it is recursively defined over the predicates by applying Boolean conjunctions, and the globally temporal operator. Satisfaction of a convex PrSTL formulae is defined as:
\begin{equation}
\label{eq:ConvexPSTL}
\begin{array}{lll}
(\xi,t) \models \lambda_{\alpha_t}^{\epsilon_t} &\!\Leftrightarrow & P(\lambda_{\alpha_t}(\xi(t)) < 0) > 1- \epsilon_t\\
(\xi,t) \models \varphi \land \psi &\Leftrightarrow & (\xi,t) \models \varphi \land (\xi,t) \models \psi\\
(\xi,t) \models \G_{[a,b]} \varphi &\Leftrightarrow & \forall t'\in [t+a, t+b], (\xi,t') \models \varphi
\end{array}
\end{equation}
\end{definition}

\begin{theorem}
\label{thm:convex}
Given a convex PrSTL formula $\varphi$, a hybrid dynamical system as in equation~\eqref{eq:dynamics}, and its initial state $x_0$. The controller synthesis problem with convex PrSTL constraints $\varphi$ defined in Problem~\ref{prob:1} is a convex program. 	
\end{theorem}

\begin{proof}
We have shown that the predicates of $\varphi$,~\ie $\lambda_{\alpha_t}^{\epsilon_t}$ create a set of convex constraints. The Boolean conjunction of convex programs are also convex; therefore, $\varphi \wedge \psi$ result in convex constraints. In addition, the \emph{globally} operator is defined as a set of finite conjunctions over its time interval: $\G_{[a,b]} \varphi = \bigwedge_{t=a}^b \varphi_t$. Thus, the \emph{globally} operator retains the convexity property of the constraints.
Consequently, Problem~\ref{prob:1}, with a convex PrSTL constraint $\varphi$ is a convex program.	
\end{proof}
Theorem~\ref{thm:convex} allows us to efficiently reduce the number of integer variables required for solving Problem~\ref{prob:1}.
We only introduce integer variables when disjunctions, eventually, or until operators appear in the \pstl~constraints.
Even when a formula is not completely part of the Convex PrSTL, integer variables are introduced only for the non-convex segments.
\begin{algorithm}
\begin{algorithmic}[1]
\Procedure{Prob. Synthesis}{$f, x_0, H,\tau, J, \varphi$}\\
\vspace{1mm}
Let $\tau = [t_1,t_2]$ is the time interval of interest.\\
\label{alg:init}
\vspace{0.5mm}
\texttt{past} $\leftarrow$ Initialize($t_1$)\\
\vspace{1mm}
\quad {\bf for }$t$ = $t_1$: $dt$: $t_2$\\
\label{alg:lin}
\vspace{0.5mm}
\quad \quad $f_{\text{lin}} = \text{linearize}(f, \xi(t))$\\
\label{alg:sense}
\vspace{0.5mm}
\quad \quad $\alpha_t \leftarrow \text{Update Distributions}(\alpha_{t-dt},\:\text{sense}(\xi_x(t))$\\
\label{alg:update}
\vspace{0.5mm}
\quad \quad $\varphi \leftarrow \varphi(\alpha_t, \epsilon_t)$ \\
\label{alg:int}
\vspace{0.5mm}
\quad \quad $\texttt{C}_\texttt{PrSTL} =  \text{MISDP}(\varphi)$\\
\vspace{0.5mm}
\quad \quad $\texttt{C} = \texttt{C}_\texttt{PrSTL} \wedge f_{\text{lin}} \wedge \:[\:\xi(t_1\cdots t-dt) = \tt{past}\:] $\\
\label{alg:opt}
\vspace{0.5mm}
\quad \quad $\uH =$ optimize$\big(J(\xi^H), \texttt{C}\big )$\\
\vspace{0.5mm}
\quad \quad $x_{t+1} = f(x_t,u_t)$\\
\vspace{0.5mm}
\quad \quad \texttt{past} $\leftarrow$ $[\texttt{past} \enskip \xi(t)]$\\
\vspace{0.5mm}
\quad {\bf end for}
\vspace{1mm}
\EndProcedure
\caption{\small{Controller Synthesis with PrSTL Formulae}}
\label{alg:1}
\end{algorithmic}
\end{algorithm}

We show our complete method of controlling dynamical systems in uncertain environments in Algorithm~\ref{alg:1}.
At the first time step $t_1$, we run an open-loop control algorithm to populate \texttt{past} in line~\ref{alg:init}.
We then run the closed-loop algorithm, finding the optimal strategy at every time step of the time interval $\tau = [t_1,t_2]$.
In the closed-loop algorithm, we linearize the dynamics at the current local state and time in line~\ref{alg:lin}, and then update the distributions over the random variables in the \pstl~formula based on new sensor data in line~\ref{alg:sense}.
Then, we update the \pstl~formulae, based on the updated distributions. If there are any other dynamic parameters that change at every time step, they can also be updated in line~\ref{alg:update}.
In line~\ref{alg:int}, we generate the mixed integer constraints in $\texttt{C}_\texttt{PrSTL}$, and then populate \texttt{C} with all the constraints including the \pstl~constraints, linearized dynamics, and enforcing the past trajectory. Note that we do not construct integer variables if the formula is in the subset of Convex PrSTL.
Then, we call the finite horizon optimization algorithm under the cost function $J(\xi^H)$, and the constraints \texttt{C} in line~\ref{alg:opt}, which provides a length $H$ strategy $\uH$. We advance the state with the first element of $\uH$, and update the history of the trajectory in~\texttt{past}.
We continue running this loop and synthesizing controllers over all time steps in interval $\tau$.

\section{Experimental Results}
\label{sec:examples}
We implemented our controller synthesis algorithm for PrSTL formulae as a Matlab toolbox, available at:\\
 \texttt{https://www.eecs.berkeley.edu/$\sim$dsadigh/PrSTL}. \\Our toolbox uses Yalmip~\cite{lofberg2004yalmip} and Gurobi~\cite{gurobi} as its optimization engine.
For all the examples we tried, the optimization computed at every step completed in less than 2 seconds on a 2.3 GHz Intel Core i7 processor with 16 GB memory.
We show some of our results for controlling quadrotors and autonomous driving under uncertain environments. 
\subsection{Quadrotor Control}
Controlling quadrotors in dynamic uncertain environments is a challenging task.
Different sources of uncertainty appear while controlling quadrotors,~\eg
uncertainty about the position of the obstacles based on classification methods, distributions over wind profiles or battery profiles, etc.
In this case study, we show how to express  properties of different models of uncertainty over time, and we find an optimal strategy under such uncertain environments.
 
We follow the derivation of the dynamics model of a quadrotor in~\cite{huang2009aerodynamics}. 
We consider a $12$ dimensional system, where the state consists of the position and velocity of the quadrotor $x,y,z$ and $\dot{x},\dot{y},\dot{z}$, as well as the Euler angles $\phi, \theta, \psi$,~\ie roll, pitch, yaw, and the angular velocities $p,q,r$. Let $\x$ be:
\begin{equation}
\label{eq:quad_state}
\x = [
x \enspace y \enspace z \enspace\dot{x} \enspace\dot{y} \enspace \dot{z}\enspace \phi \enspace\theta \enspace\psi  \enspace p \enspace q \enspace r 	
 ]^\top.
\end{equation}
The system has a $4$ dimensional control input:
 \begin{equation}
 \label{eq:quad_dynamics}
 \uu = \begin{bmatrix} u_1 & u_2 &  u_3 &  u_4\end{bmatrix}^\top,
 \end{equation}
 where $u_1$, $u_2$ and $u_3$ are the control inputs about each axis for roll, pitch and yaw respectively. $u_4$ represents the thrust input to the quadrotor in the vertical direction ($z$-axis).
The nonlinear dynamics of the system is:
\begin{equation}
\begin{aligned}
&f_1 (x,y,z) = 
	\begin{bmatrix}
		\dot{x} & 
		\dot{y} & \dot{z}
	\end{bmatrix}^\top & \\
&f_2 (\dot{x},\dot{y},\dot{z}) =
	\begin{bmatrix}
		0 & 
		0& g
			\end{bmatrix}^\top - \frac{ R_1(\dot{x},\dot{y},\dot{z})
			\begin{bmatrix}
		0 & 
		0& 0 & u_4
			\end{bmatrix}^\top}{m} &\\
			&f_3(\phi,\theta, \psi) = R_2(\dot{x},\dot{y},\dot{z})
	\begin{bmatrix}\dot{\phi}& \dot{\theta} & \dot{\psi}\end{bmatrix}^\top &\\
	&f_4 (p,q,r) = I^{-1}
	\begin{bmatrix}
		u_1 & 
		u_2& u_3
	\end{bmatrix}^\top - R_3(p,q,r) I  \begin{bmatrix}p & q & r\end{bmatrix}^\top, &
\end{aligned}
\end{equation}
 where $R_1$ and $R_2$ are rotation matrices, relating body frame and inertial frame of the quadrotor, $R_3$ is a skew-symmetric matrix, and $I$ is the inertial matrix of the rigid body. Also $g$ and $m$ denote gravity and mass of the quadrotor.
 Then the dynamics equation is:
 \begin{equation}
 f(\x,\uu) = \begin{bmatrix}
 f_1 & f_2 & f_3 & f_4	
 \end{bmatrix}^\top.
 \end{equation}.

\begin{figure}
\centering
\includegraphics[width = 0.4\textwidth]{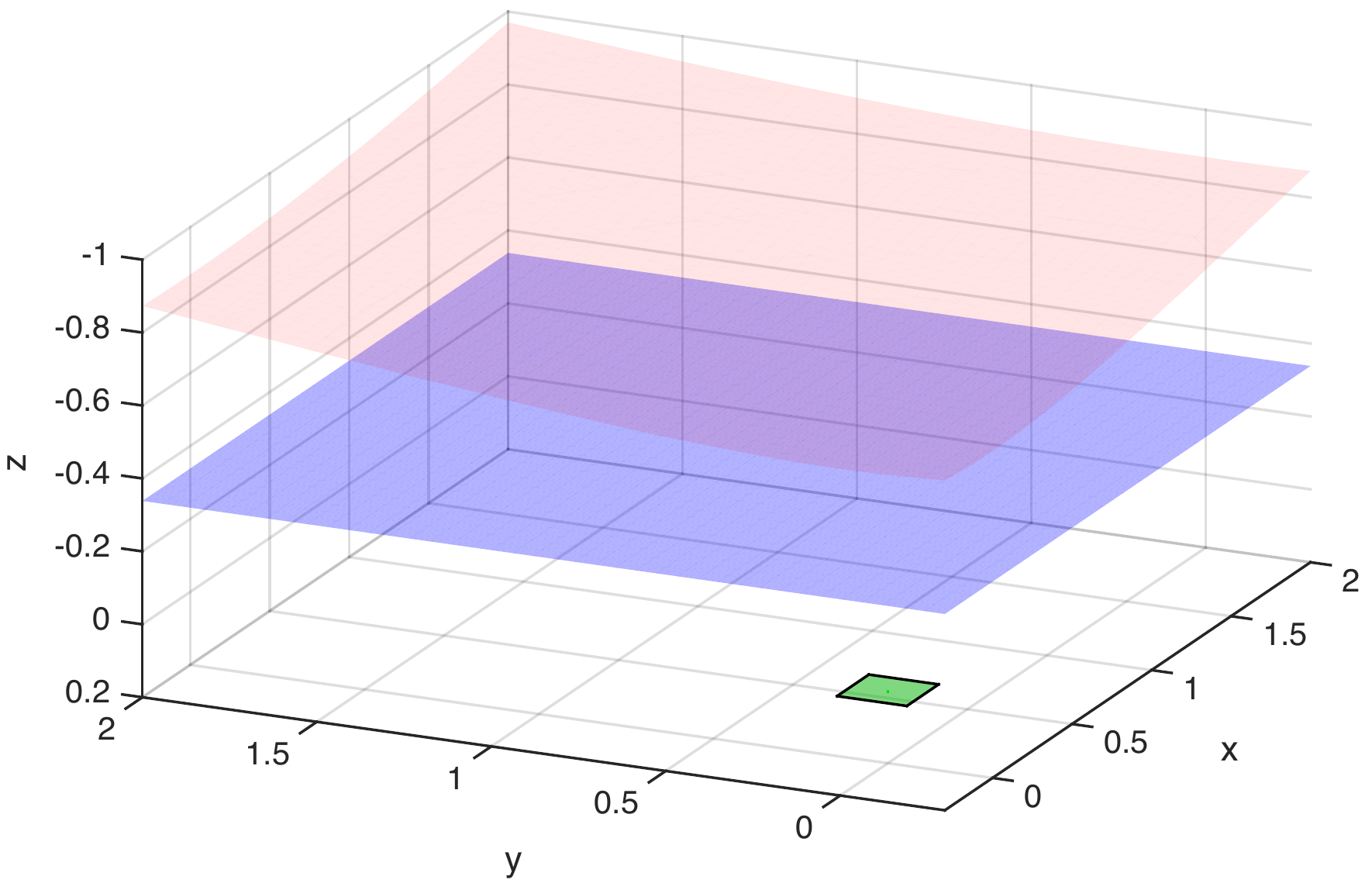}
	\caption{\small{A quadrotor starting a trajectory from the origin. The quadrotor is shown by the green square. The purple surface represents a ceiling that the quadrotor needs to avoid collisions with, and the orange surface is the quadrotors belief of where the ceiling is based on the current sensor data.}}
	\label{fig:q1}
\end{figure}

\subsubsection{Control in an Uncertain Environments}
Our first goal is for a quadrotor to reach a point in the space while avoiding obstacles. This is shown in Figure~\ref{fig:q1}, where the quadrotor is shown by a green square at its starting position, the origin $(0,0,0)$, and its objective is to reach the coordinates $(1,1,0)$ smoothly. If we let $z=0$ represent the ground level, the objective of the quadrotor is to take off and travel a distance, and then land on the ground again. Note that we use the convention, where $z <0$ is above the ground level.  We optimize the following objective:
\begin{equation}
J(\xi^H) = \sum_{t=0}^{H-1} || (x_t,y_t,z_t) - (1,1,0)||_2^2 + c||(\phi_t,\theta_t, \psi_t)||_2^2.	
\end{equation}
Here, we penalize the $l_2$-norm of the Euler angles by a factor of $c$, since we look for a smooth trajectory. We chose $c = 2$ in our examples. 
In addition to initializing the state and control input at zero, we need to satisfy the following deterministic \pstl~formulae:
\begin{equation}
\label{eq:quad_const}
\begin{aligned}
&\varphi_{\text{roll}} = \G_{[0,\infty)}(||u_1|| \leq 0.3)& \text{Bounds on Roll Input}& \\
&\varphi_{\text{pitch}} = \G_{[0,\infty)}(||u_2|| \leq 0.3)&\text{Bounds on Pitch Input}& \\
&\varphi_{\text{thrust}} = \G_{[0,\infty)}(||u_4|| \leq 0.3)&\text{Bounds on Thrust}&
\end{aligned}	
\end{equation}
In Figure~\ref{fig:q1}, the purple surface is a ceiling that the quadrotor should not collide with as it is taking off and landing at the final position.
However, the quadrotor does not have a full knowledge of where the ceiling is exactly located.
We define a sensing mechanism for the quadrotor, which consists of a meshgrid of points around the body of the quadrotor. As the system moves in the space, a Bayesian binary classifier is updated by providing a single label $-1$ (no obstacles present) or $1$ (obstacle present) for each of the sensed points.
 
The Bayesian classifier is the same as the Gaussian Process based method as described in Section~\ref{sec:Bayes} and has the linear-Gaussian form. Applying this classifier results in a Gaussian distribution for every point in the $3$D-space.
We define our classifier with confidence $1-\epsilon_t = 0.95$, as the stochastic function $\lambda^{0.05}_{\alpha_t} (\xi(t))= \alpha_t^\top [x_t \enskip y_t \enskip z_t]$, where $x_t$,$y_t$, and $z_t$ define the coordinates of the sensing points in the space, and $\alpha_t \sim \n(\mu_t,\Sigma_t)$ is the Gaussian weight inferred over time using the sensed data.
So, we define a time-varying probabilistic constraint that needs to be held at every time step as its value changes over time.
Our constraint specifies that given a classifier based on the sensing points parameterized by $\alpha_t$, we would enforce the quadrotor to stay within a safe region (defined by the classifier) with probability $1-\epsilon_t$, for $\epsilon_t = 0.05$ at all times. Thus the probabilistic formula is:
\begin{equation}
\begin{aligned}
&\varphi_{\text{classifier}} = \quad \G_{[0.1,\infty)}(\lambda^{0.05}_{\alpha_t}) \quad \text{ which is equivalent to:}& \\
&\varphi_{\text{classifier}}= \quad \G_{[0.1,\infty)}\big(P(\alpha_t^\top[x_t\enskip y_t \enskip z_t] < 0)>0.95\big)& 
\end{aligned}	
\end{equation}
We enforce this probabilistic predicate at all times in $t \in [0.1,\infty)$, which verifies the property starting from a small time after the initial state, so the quadrotor has gathered some sensor data.
In Figure~\ref{fig:q1}, the orange surface, represents the second order cone created based on $\varphi_{\text{classifier}}$, at every time step. This surface is characterized by:
\begin{equation}
\mu_t^\top \begin{bmatrix}x_t& y_t& z_t\end{bmatrix} - \Phi^{-1}(0.05) ||\Sigma_t^{1/2} \begin{bmatrix}x_t& y_t& z_t\end{bmatrix}||_2	 < 0
\end{equation}
Note that the surface shown in Figure~\ref{fig:q1}, at the initial time step is not an accurate estimate of where the ceiling is, and it is based on a distribution learned from the initial values of the sensors.
Thus, if the quadrotor was supposed to follow this estimate without updating, it would collide with the ceiling, since the orange surface showing the belief of the location of the ceiling is above the purple surface representing the real position of the ceiling.
However, the Bayesian inference running at every step of the optimization, updates the distribution over the classifier.
As shown in Figure~\ref{fig:q2}, the orange surfaces changes at every time step, since the parameters of the learned random variable $\alpha_t$, which are $\mu_t$, and $\Sigma_t$ are updated at every step.
In Figure~\ref{fig:q2}, the blue path represents the trajectory the quadrotor has already taken, and the dotted green line represents the future planned trajectory based on the current state of the classifier.
The dotted green trajectory at the initial state goes through the ceiling, since the belief of the location of the ceiling is incorrect; however, the trajectory is modified at every step as the classifier values are updated, and the quadrotor safely reaches the final position.
We solve the optimization using our toolbox, with $dt = 0.03$, and horizon length of $H = 20$. 
We emphasize that some of the constraints are time-varying, and we need to update them at every step of the optimization. We similarly update the dynamics at every time step, since we locally linearize the dynamics around the current position of the quadrotor at every step.
 \begin{figure}
 \centering
 	\includegraphics[width = 0.42\textwidth]{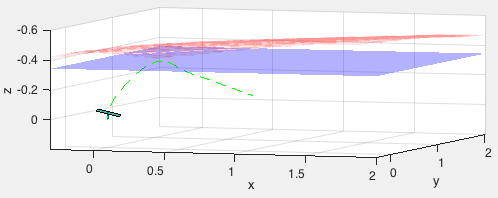}
 	\includegraphics[width = 0.42\textwidth]{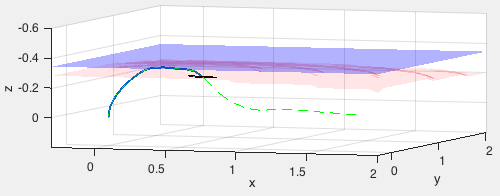}
 	\includegraphics[width = 0.42\textwidth]{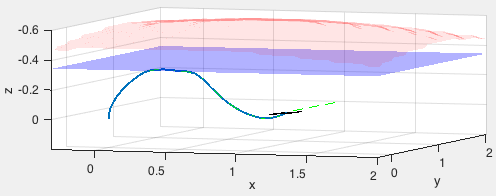}
\caption{\small{The quadrotor in Fig.~\ref{fig:q1} taking the optimal trajectory to reach the goal, while avoiding collisions with the ceiling. The figures from the top correspond to $t=0.18$ s, $t=0.63$ s, and $t=1.02$ s.}}
\label{fig:q2}
 \end{figure}
\subsubsection{Control under Battery Constraints}
We consider another scenario for controlling quadrotors, where we add a battery state to the state space of the system discussed above. So the quadrotor will be a $13$ dimensional system, where the first $12$ states follow the same order and dynamics of equations~\eqref{eq:quad_state} and~\eqref{eq:quad_dynamics}. We let $\x_t(13) = b_t$ denote the battery state, and we initialize it at $b_0 = 10$. The state of $b_t$ evolves with the negative thrust value:
\begin{equation}
f_5(b) = -|u_4|.	
\end{equation}
So the dynamics is $f = [f_1 \enskip f_2 \enskip f_3 \enskip f_4 \enskip f_5]^\top$, and all the other states and inputs are initialized at zero.
We enforce the same constraints as in equation~\eqref{eq:quad_const}, and the objective of the quadrotor is to reach the coordinates $(1,1,-0.9)$ smoothly, which corresponds to flying from the origin to the top diagonal corner of the space.
Furthermore, we impose that the quadrotor can fly above a specific height only if it is confident in its battery power.
The formula $\F_{[0,0.3]} (z_t \leq -0.1)$ encodes that eventually in the next $0.3$ s, the quadrotor will fly above a threshold level of $-0.1$.
Therefore, the truth of this formula should imply that the system is confident in the battery level, and consequently can make it to the goal position safely.
However, we assume we don't have access to the exact value of battery state due to uncertain environment factors that can affect the battery level such as radio communication, etc.
We use a stochastic linear classifier $\lambda_{\alpha_t}(\begin{bmatrix}b_t & 1\end{bmatrix}^\top)$ on a battery state augmented with value $1$, to estimate the belief on the battery level. We allow the battery state to vary with a variance $\sigma^2$ scaled at every time step of the horizon. 
So the formula ensuring that the quadrotor flies above a threshold only if its battery level is high enough is:
\begin{equation}
\begin{aligned}
\label{eq:battery}
\varphi_{\text{battery}} = \G_{[0,\infty)} \big ( \F_{[0,0.3]} (z_t \leq -0.1) \rightarrow \psi) \big),\\
\text{where} \quad
\psi =  \G_{[0,\infty)}\big(P( \alpha_t^\top\begin{bmatrix}b_t\\1\end{bmatrix}<0)\geq 1-\epsilon_t \big)\\
\alpha_t \sim \n(\begin{bmatrix}-1\\b_{\text{min}}\end{bmatrix}, \begin{bmatrix} 0& 0\\0&t\sigma^2 \end{bmatrix}),\quad \text{and}\quad \epsilon_t = 0.2.
\end{aligned}
\end{equation}
We let the confidence $1-\epsilon_t = 0.8$. The property $\psi$ can be reformulated as:
\begin{equation}
	\quad \psi = \G_{[0,\infty)}\big(P(b_t + \n(0,t\sigma^2)\geq b_{\text{min}})\geq 0.8 \big)
\end{equation}
Here $\nu \sim \n(0,1)$, and $t$ ranges over the horizon time steps.
So, $\psi$ illustrates that the quadrotor has to be confident that its battery state perturbed by a time-varying variance is above $b_{\text{min}}$ at all times.  
Therefore, $\varphi_{\text{battery}}$ specifies that if the battery state is below some threshold, the quadrotor has to fly close to the ground.
We synthesize a controller for the specifications, and the trajectory of the quadrotor is shown in Figure~\ref{fig:BAT}.
The trajectory in Figure~\ref{fig:battery} corresponds to when $\sigma = 0$,~\ie the battery state changes deterministically, and Figure~\ref{fig:battery-noise}, corresponds to $\sigma = 10$, when the quadrotor is more cautious about the state of the battery.
So the trajectory does not pass the $-0.1$ height level whenever the confidence in the battery level is below $0.8$.
\begin{figure}
\begin{subfigure}[b]{0.23\textwidth}
\centering
	\includegraphics[width = \textwidth]{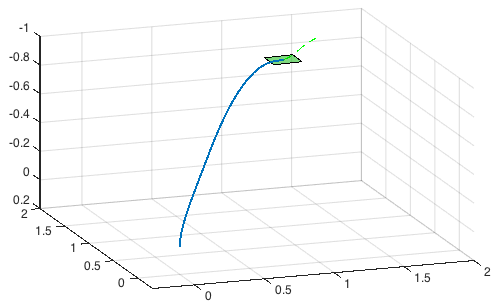}
	\caption{\small{Quadrotor with a deterministic battery state, reaching the goal, $\sigma = 0$ in equation~\eqref{eq:battery}.}}
	\label{fig:battery}
\end{subfigure}%
\hspace{10pt}
\begin{subfigure}[b]{0.23\textwidth}
\centering
	\includegraphics[width = \textwidth]{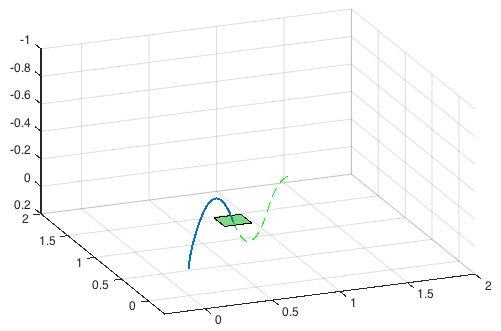}
	\caption{\small{Quadrotor with low confidence on the battery state, $\sigma = 10$ in equation~\eqref{eq:battery}. It avoids flying higher than $z=-0.1$.}}
	\label{fig:battery-noise}
\end{subfigure}
\caption{\small{Quadrotor flying to reach a goal while being confident that its battery is above some threshold.}}
\label{fig:BAT}
\end{figure}
\subsection{Autonomous Driving}
In our second case study, we consider an autonomous driving scenario. We use a simple point-mass model to define the dynamics of the vehicles on the road. We let the state of the system be $\x = [x \enskip y \enskip \theta \enskip v]^\top$, where $x$, $y$ denote the coordinates of the vehicle, $\theta$ is the heading, and $v$ is the speed. We let $\uu = [u_1 \enskip u_2]^\top$ be the control inputs, where $u_1$ is the steering input, and $u_2$ is the acceleration. Further, we write the dynamics of the vehicle as:
\begin{equation}
\label{eq:car_dynamics}
\begin{aligned}
	&\dot{x} =  v\cos(\theta)\\
	&\dot{y} = v\sin(\theta) \\
	&\dot{\theta} = \frac{v}{m} u_1\\
	&\dot{v} = u_2
\end{aligned}	
\end{equation}
\begin{figure}
\begin{subfigure}[b]{0.23\textwidth}
\centering
	\includegraphics[width = \textwidth]{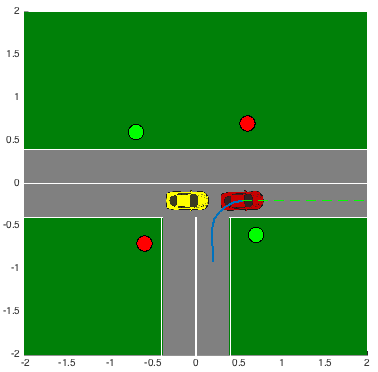}
	\caption{\small{Turning right with a deterministic model of the yellow car's velocity, $\sigma = 0$ in equation~\eqref{eq:crash}.}}
	\label{fig:car}
\end{subfigure}%
\hspace{10pt}
\begin{subfigure}[b]{0.23\textwidth}
\centering
	\includegraphics[width = \textwidth]{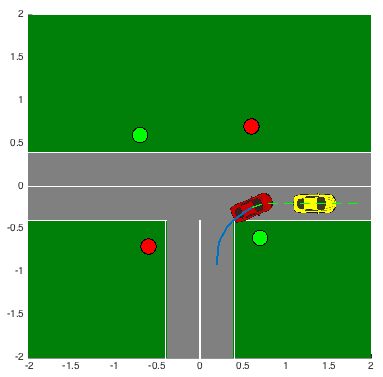}
	\caption{\small{Turning right given a probabilistic model of the yellow car's velocity, $\sigma=0.4$ in equation~\eqref{eq:crash}.}}
	\label{fig:car-noise}
\end{subfigure}
\caption{\small{Red car making a right turn at a signalized intersection. The strategy computed performs a safer trajectory with a probabilistic model of the environment, where it first waits for the yellow car to pass.}}
\label{fig:CAR}
\end{figure}
Figure~\ref{fig:CAR}, shows a scenario for an autonomous vehicle making a right turn at a signalized intersection.
We refer to the red car as the \emph{ego} vehicle,~\ie the vehicle we control autonomously, and the yellow car as the \emph{environment} vehicle.
We would like to find a strategy for the ego vehicle, so it makes a safe right turn when the traffic light is red, while yielding to the oncoming traffic.
The yellow car in the figure represents the oncoming traffic at this intersection.
In this example, the ego vehicle only has a probabilistic model of the velocity of the environment car. 
All vehicles in this example follow the same dynamics as in equation~\eqref{eq:car_dynamics}.
We refer to the states of the the ego vehicle as: $[x_{\text{}} \enskip y_{\text{}} \enskip \theta_{\text{}} \enskip v_{\text{}}]^\top$, and the states of the environment vehicle as: $[x^{\text{env}} \enskip y^{\text{env}} \enskip \theta^{\text{env}} \enskip v^{\text{env}}]^\top$. 
While synthesizing a strategy for the red car, we would enforce a set of \pstl~specifications: 
(i) We enforce bounds on control inputs and states of the two vehicles, (ii) We encode the states and transitions of the traffic light, and enforce the vehicles to obey the traffic rules, (iii) We enforce that all vehicles remain within their road boundaries.
In addition, we would like the two cars to avoid any collisions. We define collision avoidance as the following PrSTL property:
\begin{equation}
\label{eq:crash}
\begin{aligned}
&\varphi_{\text{crash}} = 	\G_{[0, \infty)}\Big(&\\
& \quad \enskip P \big(x_t - (x_t^{\text{env}}+ s_{t,x}t) \geq \delta \big) \enskip \geq 1-\epsilon_t &\\
&\vee \enskip P \big(x_t - (x_t^{\text{env}}+ s_{t,x}t) \leq -\delta \big) \enskip \geq 1-\epsilon_t &\\
&\vee \enskip P \big(y_t - (y_t^{\text{env}}+ s_{t,y}t) \geq \delta \big) \enskip \geq 1-\epsilon_t & \\
&\vee \enskip P \big(y_t - (y_t^{\text{env}}+ s_{t,y}t) \leq -\delta \big) \enskip \geq 1-\epsilon_t &\Big) 
\end{aligned}
\end{equation}
\begin{equation*}
\begin{aligned}
\begin{bmatrix}
s_{t,x}\\
s_{t,y}	
\end{bmatrix} =\n(v_t^{\text{env}},\sigma^2)
\begin{bmatrix}
\cos(\theta_t^{\text{env}})\\
\sin(\theta_t^{\text{env}}	)
\end{bmatrix} \quad \text{and} \quad \delta = 0.4, \enskip \epsilon_t = 0.2.
\end{aligned}	
\end{equation*}
Here, $\varphi_{\text{crash}}$ consists of a global operator at all times over the disjunction of four PrSTL predicates. Each probabilistic predicate encodes possible crash between the two vehicles.
In equation~\eqref{eq:crash}, $\delta$ represents the minimum distance between the $x$ and $y$ coordinates of the two vehicles in either direction,
 which generates the four disjunctions on the predicates.
The estimate of the distance between $x$ and $y$ coordinates of the two vehicles is encoded in each predicate, by considering the difference between the coordinates of the ego vehicle, and the propagated coordinates of the environment vehicle based on the value of its velocity. 
The velocity is a vector of Gaussian random variables $\begin{bmatrix}
s_{t,x}&
s_{t,y}	
\end{bmatrix}^\top$ computed based on the current heading of the environment vehicle $\theta_t^{\text{env}}$, centered at the current speed $v^\text{env}_t$, and perturbed by a variance $\sigma^2$.
The predicates in $\varphi_{\text{crash}}$ define a linear classifier on the signal representing the coordinates of the ego vehicle, parameterized by a random variable characterizing the velocity of the environment vehicle.
These predicates can easily be reformulated to the nominal structure of a PrSTL predicate $\lambda_{\alpha_t}^{\epsilon_t}$. However, we leave them as in equation~\eqref{eq:crash} for better illustration.

In the autonomous driving example, we use a sampling time of $dt = 0.1$ s, and horizon of $H=20$. In addition, we let $\sigma = 0.4$.
We successfully synthesize a strategy for the autonomous vehicle by solving Problem~\ref{prob:1}, and following the steps in Algorithm~\ref{alg:1}. The trajectory generated using this strategy is shown by the solid blue line in Figure~\ref{fig:CAR}. The dotted green line is the future trajectory computed by the MPC scheme.
In Figure~\ref{fig:car}, the ego vehicle has a deterministic model of the environment vehicle as $\sigma=0$; therefore, it performs the right turn before letting the environment vehicle pass.
However, as shown in Figure~\ref{fig:car-noise}, given $\sigma=0.4$, and $\epsilon_t = 0.2$, the ego vehicle is not confident enough in avoiding collisions, so it acts in a conservative manner and waits for the environment car to pass first, and then performs its right turn safely. 

\section{Conclusion}
\label{sec:dis}
We have presented a framework for safe controller synthesis under uncertainty. The key contributions include defining PrSTL, a logic for expressing probabilistic properties that allows embedding Bayesian graphical models. We also show how to synthesize control in a receding horizon framework under PrSTL specifications that express Bayesian linear classifiers. Another distinguishing aspect of this work is that the resulting logic adapts as more data is observed with the evolution of the system. We demonstrate the effectiveness of the approach by synthesizing safe strategies for a quadrotor and an autonomous vehicle traveling in uncertain environments.
 
The presented approach extends easily to distributions other than Gaussians via Bayesian approximate inference techniques~\cite{minka2001family,beal2003variational} that can project distributions to the Gaussian densities. Future work includes, extending controller synthesis for arbitrary distributions via sampling based approaches; we are also exploring using the proposed framework as a building block for complex robotic tasks that need to invoke higher level planning algorithms.
\newpage

%

\bibliographystyle{abbrv}
\small{
\bibliography{refs}

\begin{thebibliography}{10}

\bibitem{akametalu2014reachability}
A.~K. Akametalu, S.~Kaynama, J.~F. Fisac, M.~N. Zeilinger, J.~H. Gillula, and
  C.~J. Tomlin.
\newblock Reachability-based safe learning with gaussian processes.
\newblock In {\em 2014 IEEE 53rd Annual Conference on Decision and Control}.

\bibitem{andrieu2003introduction}
C.~Andrieu, N.~De~Freitas, A.~Doucet, and M.~I. Jordan.
\newblock An introduction to mcmc for machine learning.
\newblock {\em Machine learning}, 50(1-2):5--43, 2003.

\bibitem{aswani2013provably}
A.~Aswani, H.~Gonzalez, S.~S. Sastry, and C.~Tomlin.
\newblock Provably safe and robust learning-based model predictive control.
\newblock {\em Automatica}, 49(5):1216--1226, 2013.

\bibitem{beal2003variational}
M.~J. Beal.
\newblock {\em Variational algorithms for approximate Bayesian inference}.
\newblock University of London, 2003.

\bibitem{ben2009robust}
A.~Ben-Tal, L.~El~Ghaoui, and A.~Nemirovski.
\newblock {\em Robust optimization}.
\newblock Princeton University Press, 2009.

\bibitem{blackmore2011chance}
L.~Blackmore, M.~Ono, and B.~C. Williams.
\newblock Chance-constrained optimal path planning with obstacles.
\newblock {\em IEEE Transactions on Robotics}, 27(6):1080--1094, 2011.

\bibitem{boyd2004convex}
S.~Boyd and L.~Vandenberghe.
\newblock {\em Convex optimization}.
\newblock Cambridge university press, 2004.

\bibitem{carvalhoa2014stochastic}
A.~Carvalho, Y.~Gao, S.~Lefevre, and F.~Borrelli.
\newblock Stochastic predictive control of autonomous vehicles in uncertain
  environments.
\newblock In {\em 12th International Symposium on Advanced Vehicle Control},
  2014.

\bibitem{donze2013efficient}
A.~Donz{\'e}, T.~Ferrere, and O.~Maler.
\newblock Efficient robust monitoring for {STL}.
\newblock In {\em Computer Aided Verification}. Springer, 2013.

\bibitem{donze2010robust}
A.~Donz{\'e} and O.~Maler.
\newblock {\em Robust Satisfaction of Temporal Logic over Real-Valued Signals}.
\newblock 2010.

\bibitem{fu2014integrating}
J.~Fu and U.~Topcu.
\newblock Integrating active sensing into reactive synthesis with temporal
  logic constraints under partial observations.
\newblock {\em arXiv:1410.0083}, 2014.

\bibitem{fu2014probably}
J.~Fu and U.~Topcu.
\newblock Probably approximately correct mdp learning and control with temporal
  logic constraints.
\newblock {\em arXiv:1404.7073}, 2014.

\bibitem{fu2015computational}
J.~Fu and U.~Topcu.
\newblock Computational methods for stochastic control with metric interval
  temporal logic specifications.
\newblock {\em arXiv:1503.07193}, 2015.

\bibitem{gelman2014bayesian}
A.~Gelman, J.~B. Carlin, H.~S. Stern, and D.~B. Rubin.
\newblock {\em Bayesian data analysis}, volume~2.
\newblock Taylor \& Francis, 2014.

\bibitem{gillula2012guaranteed}
J.~H. Gillula and C.~J. Tomlin.
\newblock Guaranteed safe online learning via reachability: tracking a ground
  target using a quadrotor.
\newblock In {\em 2012 IEEE International Conference on Robotics and
  Automation}.

\bibitem{gol2015temporal}
E.~A. Gol, M.~Lazar, and C.~Belta.
\newblock Temporal logic model predictive control.
\newblock {\em Automatica}, 56:78--85, 2015.

\bibitem{griva2009linear}
I.~Griva, S.~G. Nash, and A.~Sofer.
\newblock {\em Linear and nonlinear optimization}.
\newblock Siam, 2009.

\bibitem{gurobi}
I.~Gurobi~Optimization.
\newblock Gurobi optimizer reference manual, 2015.

\bibitem{huang2009aerodynamics}
H.~Huang, G.~M. Hoffmann, S.~L. Waslander, and C.~J. Tomlin.
\newblock Aerodynamics and control of autonomous quadrotor helicopters in
  aggressive maneuvering.
\newblock In {\em 2009 IEEE International Conference on Robotics and
  Automation}.

\bibitem{jordan1998learning}
M.~Jordan.
\newblock Learning in graphical models (adaptive computation and machine
  learning).
\newblock 1998.

\bibitem{kataoka1963stochastic}
S.~Kataoka.
\newblock A stochastic programming model.
\newblock {\em Econometrica: Journal of the Econometric Society}, 1963.

\bibitem{kothare1996robust}
M.~V. Kothare, V.~Balakrishnan, and M.~Morari.
\newblock Robust constrained model predictive control using linear matrix
  inequalities.
\newblock {\em Automatica}, 32(10):1361--1379, 1996.

\bibitem{kress2009temporal}
H.~Kress-Gazit, G.~E. Fainekos, and G.~J. Pappas.
\newblock Temporal-logic-based reactive mission and motion planning.
\newblock {\em IEEE Transactions on Robotics}, 25(6):1370--1381, 2009.

\bibitem{lenz2015stochastic}
D.~Lenz, T.~Kessler, and A.~Knoll.
\newblock Stochastic model predictive controller with chance constraints for
  comfortable and safe driving behavior of autonomous vehicles.
\newblock In {\em Intelligent Vehicles Symposium (IV), 2015 IEEE}.

\bibitem{livingston2012backtracking}
S.~C. Livingston, R.~M. Murray, and J.~W. Burdick.
\newblock Backtracking temporal logic synthesis for uncertain environments.
\newblock In {\em 2012 IEEE International Conference on Robotics and
  Automation}.

\bibitem{lofberg2004yalmip}
J.~L{\"o}fberg.
\newblock Yalmip: A toolbox for modeling and optimization in matlab.
\newblock In {\em 2004 IEEE International Symposium on Computer Aided Control
  Systems Design}.

\bibitem{maler2004monitoring}
O.~Maler and D.~Nickovic.
\newblock Monitoring temporal properties of continuous signals.
\newblock In {\em Formal Techniques, Modelling and Analysis of Timed and
  Fault-Tolerant Systems}. Springer.

\bibitem{minka2001family}
T.~P. Minka.
\newblock {\em A family of algorithms for approximate Bayesian inference}.
\newblock PhD thesis, Massachusetts Institute of Technology, 2001.

\bibitem{mitchell2000level}
I.~Mitchell and C.~J. Tomlin.
\newblock Level set methods for computation in hybrid systems.
\newblock In {\em Hybrid Systems: Computation and Control}. Springer.

\bibitem{mitchell2005time}
I.~M. Mitchell, A.~M. Bayen, and C.~J. Tomlin.
\newblock A time-dependent hamilton-jacobi formulation of reachable sets for
  continuous dynamic games.
\newblock {\em IEEE Transactions on Automatic Control}, 50, 2005.

\bibitem{morari1993model}
M.~Morari, C.~Garcia, J.~Lee, and D.~Prett.
\newblock {\em Model predictive control}.
\newblock Prentice Hall Englewood Cliffs, NJ, 1993.

\bibitem{piterman2006synthesis}
N.~Piterman, A.~Pnueli, and Y.~Sa$'$ar.
\newblock Synthesis of reactive (1) designs.
\newblock In {\em Verification, Model Checking, and Abstract Interpretation},
  pages 364--380. Springer, 2006.

\bibitem{puggelli2013polynomial}
A.~Puggelli, W.~Li, A.~L. Sangiovanni-Vincentelli, and S.~A. Seshia.
\newblock Polynomial-time verification of pctl properties of mdps with convex
  uncertainties.
\newblock In {\em Computer Aided Verification}. Springer, 2013.

\bibitem{raman2014model}
V.~Raman, A.~Donz{\'e}, M.~Maasoumy, R.~M. Murray, A.~Sangiovanni-Vincentelli,
  S.~Seshia, et~al.
\newblock Model predictive control with signal temporal logic specifications.
\newblock In {\em 2014 IEEE 53rd Annual Conference on Decision and Control}.

\bibitem{Raman15}
V.~Raman, A.~Donz{\'e}, D.~Sadigh, R.~M. Murray, and S.~A. Seshia.
\newblock Reactive synthesis from signal temporal logic specifications.
\newblock In {\em 18th International Conference on Hybrid Systems: Computation
  and Control}, 2015.

\bibitem{sadigh2014learning}
D.~Sadigh, E.~S. Kim, S.~Coogan, S.~S. Sastry, S.~Seshia, et~al.
\newblock A learning based approach to control synthesis of markov decision
  processes for linear temporal logic specifications.
\newblock In {\em 2014 IEEE 53rd Annual Conference on Decision and Control}.

\bibitem{svorevnova2015temporal}
M.~Svore{\v{n}}ov{\'a}, J.~K{\v{r}}et{\'\i}nsk{\`y}, M.~Chmel{\'\i}k,
  K.~Chatterjee, I.~{\v{C}}ern{\'a}, and C.~Belta.
\newblock Temporal logic control for stochastic linear systems using
  abstraction refinement of probabilistic games.
\newblock In {\em 18th International Conference on Hybrid Systems: Computation
  and Control}, 2015.

\bibitem{van1963minimum}
C.~Van~de Panne and W.~Popp.
\newblock Minimum-cost cattle feed under probabilistic protein constraints.
\newblock {\em Management Science}, 9(3):405--430, 1963.

\bibitem{vitus2012stochastic}
M.~P. Vitus.
\newblock {\em Stochastic Control Via Chance Constrained Optimization and its
  Application to Unmanned Aerial Vehicles}.
\newblock PhD thesis, Stanford University, 2012.

\bibitem{vitus2013probabilistic}
M.~P. Vitus and C.~J. Tomlin.
\newblock A probabilistic approach to planning and control in autonomous urban
  driving.
\newblock In {\em 2013 IEEE 52nd Annual Conference on Decision and Control}.

\bibitem{wang1992robust}
Y.~Wang, L.~Xie, and C.~E. de~Souza.
\newblock Robust control of a class of uncertain nonlinear systems.
\newblock {\em Systems \& Control Letters}, 19(2):139--149, 1992.

\bibitem{williams2006gaussian}
C.~K. Williams and C.~E. Rasmussen.
\newblock Gaussian processes for machine learning.
\newblock {\em the MIT Press}, 2(3):4, 2006.

\bibitem{wongpiromsarn2010receding}
T.~Wongpiromsarn, U.~Topcu, and R.~M. Murray.
\newblock Receding horizon control for temporal logic specifications.
\newblock In {\em 13th ACM international conference on Hybrid systems:
  computation and control}, 2010.

\bibitem{zhou1998essentials}
K.~Zhou and J.~C. Doyle.
\newblock {\em Essentials of robust control}, volume 180.
\newblock Prentice hall Upper Saddle River, NJ, 1998.

\end{thebibliography}
}
\end{document}